\documentclass{amsart}
\usepackage{amssymb,graphicx}
\usepackage{mathrsfs}
\usepackage{color}
\usepackage{enumerate}
\usepackage{a4wide}
\usepackage[colorlinks=true,citecolor=blue,linkcolor=blue]{hyperref}

\newcommand{\Rl}{\mathbb{R}}
\newcommand{\mbR}{\mathbb{R}}

\newcommand{\Lc}{\mathcal{L}}

\newcommand{\Tr}{\mathrm{Tr}}

\newcommand{\nexp}[1]{\exp(#1)}

\def\XXint#1#2#3{{\setbox0=\hbox{$#1{#2#3}{\int}$ }
\vcenter{\hbox{$#2#3$ }}\kern-.6\wd0}}

\numberwithin{equation}{section}

\newtheorem{theorem}{Theorem}[section]

\newtheorem{lemma}[theorem]{Lemma}

\theoremstyle{remark}

\title{The density of states depends on the domain}
\author{N.\,Azamov, E.\,McDonald, D.\,Zanin and F.\,Sukochev}

\date{\today}

\begin{document}

\begin{abstract} In this short note we demonstrate that the definition of the density of states of a Schr\"odinger operator with bounded potential in general depends on the choice of the domain undergoing the thermodynamic limit. 
\end{abstract}
\maketitle{}

\section{Introduction}

The density of states of a Schr\"odinger operator is defined as the averaged number of states per unit of volume and is a quantity of great importance in condensed matter physics. The Laplace transform of the density of states is the partition function which often serves as the initial point of theoretical investigation of a statistical mechanical system, see e.g. \cite{Ruelle}. Its definition involves restriction of the system to a sequence of finite volume domains $\{\Omega_n\}_{n=0}^\infty$ such that $\Omega_n\to \Rl^d$ and taking the thermodynamic limit $n\to\infty.$ This process involves ambiguities in the choice of boundary conditions and of the shape of the domains $\{\Omega_n\}_{n=0}^\infty.$ It is well-known that the sequence $\{\Omega_n\}_{n=0}^\infty$ must at least satisfy some F\o{}lner condition for the limit to be meaningful.

It has been known for some time that in general the density of states is independent of the boundary conditions \cite{DIT2001}, \cite[Theorem C.7.4]{SimonSemigroups}. 
Most of the literature is devoted to potentials which are periodic or almost periodic, or random potentials satisfying some ergodicity condition. In these cases it is usually the case that the choice of F\o{}lner sequence $\{\Omega_n\}_{n=0}^\infty$ is irrelevant \cite{AiWa,Beresin_Shubin,CarmonaLacroix}. The result of this note is that in general the choice of approximating F\o{}lner sequence does matter.

We prove that there exist potentials $V\in L_{\infty}(\Rl^d)$
such that the density of states defined with the sequence of balls $\{B(0,n)\}_{n=0}^\infty$ differs from that defined by the sequence of boxes $\{[-n,n]^d\}_{n=0}^\infty$. 
The potentials which we use for this purpose are radially homogeneous. That is, the ones which obey the condition $V(tx) = V(x),$ $x \in \mbR^d,  t > 0.$ Besides our paper \cite{AMSZ}, a handful of articles concerning Schr\"odinger operators with radially homogeneous potentials exist \cite{Herb,HSk1,HSk2}. The idea of the proof is that the thermodynamic limit $\Omega_n \to \mbR^d$ for radially homogeneous potentials corresponds to the semi-classical limit $\hbar \to 0.$ 


\section{Proofs}
For a bounded open subset $0 \in \Omega$ of $\Rl^d$, we denote by $\Delta_{\Omega}$ the Laplace operator in~$\Omega$ with Dirichlet boundary conditions, and by $|\Omega|$ we denote the Lebesgue measure of $\Omega$. Let $L_2(\Omega)$ denote
the Hilbert space of almost-everywhere equivalence classes of square integrable functions on $\Omega$ with the Lebesgue measure. Denote by $\Tr_{L_2(\Omega)}$ the operator trace on the ideal $\Lc_1(L_2(\Omega))$ of trace class
operators on $L_2(\Omega).$
Let~$V$ be a bounded measurable real-valued function on $\Rl^d$. We denote by $M_V$ the operator on $L_2(\Rl^d)$ of pointwise multiplication by $V$.
It is known that if for all $t > 0$ there exists the limit
\begin{equation} \label{F: thermo lim = DOS}
     \lim_{R\to\infty} \frac{1}{|R\Omega|}\Tr_{L_2(R\Omega)}(\nexp{-t(-\Delta_{R\Omega}+M_V)})
\end{equation}
then there exists a unique measure $\nu_{V,\Omega}$ on $\Rl$ such that
\begin{equation*}
    \lim_{R\to\infty} \frac{1}{|R\Omega|}\Tr_{L_2(R\Omega)}(\nexp{-t(-\Delta_{R\Omega}+M_V)}) = \int_{\Rl} \nexp{-t\lambda} \,d\nu_{V,\Omega}(\lambda),\quad t > 0.
\end{equation*}
See \cite[Proposition C.7.2]{SimonSemigroups}. The measure $\nu_{V,\Omega}$ is called the \emph{density of states}.

\begin{theorem}\label{main_result}
    Let $V \in L_\infty(\Rl^d)$ be a radially homogeneous potential, that is, for all $t > 0$ and for all $x\in \Rl^d$ we have $V(tx) = V(x)$. Then
        for every bounded open set $\Omega$ containing zero and with piecewise smooth boundary, the density of states $\nu_{V,\Omega}$ exists and is given by the formula
        \begin{equation*}
             \int_{\Rl} \nexp{-t\lambda} \,d\nu_{V,\Omega}(\lambda) = (4\pi t)^{-\frac{d}{2}}\frac{1}{|\Omega|}\int_{\Omega} \nexp{-tV(x)}\,dx,\quad t > 0.
        \end{equation*}
\end{theorem}
The above formula for the Laplace transform of $\nu_{V,\Omega}$ yields the following for the integrated density of states function $\lambda\mapsto\nu_{V,\Omega}(-\infty,\lambda],$
\begin{equation*}
    \nu_{V,\Omega}(-\infty,\lambda] = \frac{1}{|\partial \Omega|}\int_{\partial \Omega} \nu_{V(\sigma),\Omega}(-\infty,\lambda]\,d\sigma.
\end{equation*}
Here, $|\partial \Omega|$ is the $(d-1)$-Hausdorff measure of $\partial \Omega$ and $d\sigma$ is the corresponding measure, and $V(\sigma)$ is the value of $V$ at the point $\sigma\in \partial \Omega.$
This is a generalisation of the formula in \cite[Theorem 2.1]{AMSZ} to arbitrary domains, but Theorem \ref{main_result} is stronger in that it also proves existence of the density of states.

We give a proof of Theorem \ref{main_result} based on the following well-known semiclassical Weyl law, see e.g. \cite[Theorem C.6.2]{SimonSemigroups},
\cite[Theorem 10.1]{simon_functional}. For the special case of smooth $V$, see \cite[Chapter 6]{Zw}.
\begin{theorem}\label{semiclassical_weyl_law}
    Let $\Omega\subset \Rl^d$ be a bounded open set with piecewise smooth boundary, and let $V \in L_{\infty}(\Omega)$. Then
    \begin{equation*}
        \lim_{\hbar\to 0} \hbar^d\Tr_{L_2(\Omega)}(\nexp{-t(-\hbar^2\Delta_{\Omega}+M_V)}) = (4\pi t)^{-\frac{d}{2}}\int_{\Omega} \nexp{-tV(x)}\,dx.
    \end{equation*}
\end{theorem}

The following lemma is a key observation.
\begin{lemma}\label{rescaling_lemma}
    Let $R > 0$, and let $0 \in \Omega\subset \Rl^d$ be an open set with piecewise smooth boundary. If $V \in L_{\infty}(\Rl^d)$ is homogeneous, then
    \begin{equation*}
        \Tr_{L_2(R\Omega)}(\nexp{-t(-\Delta_{R\Omega}+M_{V})}) = \Tr_{L_2(\Omega)}(\nexp{-t(-R^{-2}\Delta_{\Omega}+M_V)}).
    \end{equation*}
\end{lemma}
\begin{proof}
    Let $U_R$ denote the unitary mapping
        $U_R:L_2(R\Omega)\to L_2(\Omega)$
    given by
    \begin{equation*}
        (U_Ru)(x) = R^{\frac{d}{2}}u(Rx),\quad x \in \Omega.
    \end{equation*}
    For any operator $T \in \Lc_1(L_2(R\Omega))$, we have
    \begin{equation}\label{unitary_trace_conjugation}
      \Tr_{L_2(R\Omega)}(T) = \Tr_{L_2(\Omega)}(U_RTU_R^*). 
    \end{equation}
    Since  $ U_RM_VU_R^* u(x) = V(Rx)u(x)$ and since $V$ is radially homogeneous
     we have
      $ M_V = U_RM_VU_R^*$
    (on the left hand side, $M_V$ is understood as acting on $L_2(\Omega)$ and on the right it acts on $L_2(R\Omega)$).
Combining this with 
      $  U_R\Delta_{R\Omega}U_R^* 
           = R^{-2}\Delta_{\Omega},$
we conclude that 
    \begin{align*}
        \Tr_{L_2(R\Omega)}(\nexp{-t(-\Delta_{R\Omega}+M_{V})}) &\stackrel{\text{\eqref{unitary_trace_conjugation}}}{=} \Tr_{L_2(\Omega)}(U_R\nexp{-t(-\Delta_{R\Omega}+M_{V})}U_R^*)\\
                                                               &= \Tr_{L_2(\Omega)}(\nexp{-t(-R^{-2}\Delta_{\Omega}+M_V)}).
    \end{align*}
\end{proof}

\begin{proof}[Proof of Theorem \ref{main_result}]
    By Lemma \ref{rescaling_lemma}, we have
    \begin{equation*}
        \frac{1}{|R\Omega|}\Tr_{L_2(R\Omega)}(\nexp{-t(-\Delta_{R\Omega}+M_V)}) = \frac{1}{|\Omega|}R^{-d}\Tr_{L_2(\Omega)}(\nexp{-t(-R^{-2}\Delta_{\Omega}+M_V)})
    \end{equation*}
    Let $\hbar = R^{-1}$, so that
    \begin{equation*}
        \frac{1}{|R\Omega|}\Tr_{L_2(R\Omega)}(\nexp{-t(-\Delta_{R\Omega}+M_V)}) = \frac{1}{|\Omega|} \hbar^d \Tr_{L_2(\Omega)}(\nexp{-t(-\hbar^2\Delta_{\Omega}+M_V)}).
    \end{equation*}
    According to Theorem \ref{semiclassical_weyl_law}, the limit as $\hbar\to 0$ (equivalently, as $R\to\infty$) exists, and 
    \begin{align*}
        \lim_{R\to\infty} \frac{1}{|R\Omega|}\Tr_{L_2(R\Omega)}(\nexp{-t(-\Delta_{R\Omega}+M_V)}) &= \frac{1}{|\Omega|}\lim_{\hbar\to 0}\hbar^d \Tr_{L_2(\Omega)}(\nexp{-t(-\hbar^2\Delta_{\Omega}+M_V)})\\
                                                                                               &= \frac{1}{|\Omega|}(4\pi t)^{-\frac{d}{2}}\int_{\Omega} \nexp{-tV(x)}\,dx.
    \end{align*}
    Hence in this case the limit in \eqref{F: thermo lim = DOS} exists, and we deduce the existence of a density of states measure $\nu_{V,\Omega}.$ The above computation yields the formula
    \begin{equation*}
        \int_{\Rl} \nexp{-t\lambda}\,d\nu_{V,\Omega}(\lambda) = \frac{1}{|\Omega|}(4\pi t)^{-\frac{d}{2}}\int_{\Omega} \nexp{-tV(x)}\,dx,\quad t > 0. 
    \end{equation*}
\end{proof}

\begin{theorem} 
There exist radially homogeneous potentials $V\in L_{\infty}(\Rl^d)$ for which 
          $  \nu_{V,[-1,1]^d} \neq \nu_{V,B(0,1)},  $
        where $B(0,1)$ is the open ball of radius $1$ in $\Rl^d.$ 
\end{theorem}
\begin{proof} 
    It follows from Theorem \ref{main_result} that for any bounded open set $\Omega$ containing zero and with piecewise smooth boundary we have
    \begin{equation*}
        \int_{-\infty}^\infty \nexp{-t\lambda} \,d\nu_{V,\Omega}(\lambda) = (4\pi t)^{-\frac{d}{2}}  \left(1+\frac{t}{|\Omega|}\int_{\Omega} V(x)\,dx+O(t^2)  \right),\quad t\to 0.
    \end{equation*}
    Hence, to prove the claim it suffices to give an example of a radially homogeneous potential $V \in L_{\infty}(\Rl^d)$ such that
    \begin{equation*}
        \frac{1}{|[-1,1]^d|}\int_{[-1,1]^d} V(x)\,dx \neq \frac{1}{|B(0,1)|} \int_{B(0,1)} V(x)\,dx.
    \end{equation*}
    Such an example is given by the function
    \begin{equation*}
       V(x) = \frac{|x_1x_2|}{|x_1|^2+|x_2|^2},\quad x = (x_1,\ldots,x_d) \in \Rl^d.
    \end{equation*}
    Indeed, for $d=2$ we may compute
    \begin{equation*}
        \frac{1}{|[-1,1]^2|}\int_{[-1,1]^2} V(x)\,dx = \int_0^1\int_0^1 \frac{x_1x_2}{x_1^2+x_2^2}\,dx_1dx_2 = \frac{1}{2}\log(2)
    \end{equation*}
    and
    \begin{equation*}
        \frac{1}{|B(0,1)|}\int_{B(0,1)} V(x)\,dx = \frac{4}{\pi} \int_0^{\frac{\pi}{2}}\int_0^1 r\cos(\theta)\sin(\theta)\, drd\theta = \frac{1}{\pi}.
    \end{equation*}
    For $d>2$ the computation is identical.
\end{proof}

\end{document}